\documentclass{article}
\usepackage[utf8]{inputenc}
\usepackage[T2A]{fontenc}
\usepackage[english]{babel}
\inputencoding{utf8}

\usepackage{amsmath,amssymb,amsthm, nicefrac}
\usepackage{multicol}

\usepackage[
    pdfpagemode=UseNone,
    bookmarks=true,
    bookmarksopen=true,
    pdfpagelayout=SinglePage,
    pdfstartview={XYZ null null 1},
    pdfhighlight=/P,
    colorlinks=true,
    linkcolor=blue,
    citecolor=blue,
    urlcolor=blue
    ]{hyperref}
\usepackage{bookmark}
\bookmarksetup{
  open,
  numbered,
  depth=5, 
}
\usepackage{blkarray}
\usepackage{float}

\advance\textwidth 20mm  \advance\oddsidemargin -10mm
\advance\textheight 20mm \advance\topmargin -15mm

\def\C{\mathbb{C}}

\theoremstyle{plain}
\newtheorem{theorem}{Theorem}
\newtheorem{lemma}[theorem]{Lemma}
\newtheorem{proposition}{Proposition}
\newtheorem{corollary}{Corollary}
\newtheorem{definition}{Definition}
\newtheorem{assumption}{Assumption}
\newtheorem{remark}{Remark}
\newtheorem{example}{Example}

\DeclareMathOperator{\Mat}{Mat}

\DeclareMathOperator{\trace}{trace}

\newcommand{\brackets}[1]{\left( #1 \right)}

\DeclareMathOperator{\PI}{P_{1}}
\DeclareMathOperator{\PII}{P_{2}}

\DeclareMathOperator{\PIIIpr}{P_{3}^{\prime}}
\DeclareMathOperator{\PIV}{P_{4}}
\DeclareMathOperator{\PV}{P_{5}}
\DeclareMathOperator{\PVI}{P_{6}}

\newcommand{\PVIn}[1]{\text{P}_6^{ #1 }}
\newcommand{\PVn}[1]{\text{P}_5^{ #1 }}
\newcommand{\PIVn}[1]{\text{P}_4^{ #1 }}
\newcommand{\PIIIprn}[1]{\text{P}_3^{\prime \, #1 }}

\newcommand{\PIIn}[1]{\text{P}_2^{ #1 }}
\newcommand{\PIn}[1]{\text{P}_1^{ #1 }}

\newcommand{\HI}{H_{1}}
\newcommand{\HII}{H_{2}}
\newcommand{\HIII}{H_{3}}
\newcommand{\HIV}{H_{4}}
\newcommand{\HV}{H_{5}}
\newcommand{\HVI}{H_{6}}

\newcommand{\Painleve}{Painlev{\'e} }
\newcommand{\PPainleve}{Painlev{\'e}}

\makeatletter
\newcounter{savesection}
\newcounter{apdxsection}
\renewcommand\appendix{\par
  \setcounter{savesection}{\value{section}}%
  \setcounter{section}{\value{apdxsection}}%
  \setcounter{subsection}{0}%
  \gdef\thesection{\@Alph\c@section}}
\newcommand\unappendix{\par
  \setcounter{apdxsection}{\value{section}}%
  \setcounter{section}{\value{savesection}}%
  \setcounter{subsection}{0}%
  \gdef\thesection{\@arabic\c@section}}
\makeatother

\usepackage{mathtools}
\mathtoolsset{showonlyrefs,showmanualtags}

\allowdisplaybreaks

\title{Classification of Hamiltonian non-abelian Painlevé~type~systems}

\date{}

\author{I.A. Bobrova\thanks{National Research Univerisity Higher School of Economics, Moscow, Russian Federation.},~ V.V. Sokolov\thanks{L.D.~Landau Institute for Theoretical Physics, Chernogolovka, Russian Federation.} }

\usepackage{tikz}

\begin{document}
\maketitle

\begin{abstract}
All Hamiltonian non-abelian \Painleve systems of $\PI-\PVI$ type with constant coefficients are found. For $\PI-\PV$ systems, we replace an appropriate inessential constant parameter with a non-abelian constant.
To prove the integrability of new $\PIIIpr$ and $\PV$ systems thus obtained, we find isomonodromic Lax pairs for them. 
\medskip

\end{abstract}

\section{Introduction}
\hspace{3.8mm} In the paper \cite{okamoto1980polynomial} by K.Okamoto all \Painleve equations $\PI-\PVI$ were written as polynomial Hamiltonian systems of the form 
\begin{gather} \label{Ham}
    \left\{
    \begin{array}{lcl}
       \displaystyle \frac{d u}{d z} 
         &=&\displaystyle \frac{\partial H}{\partial v},
         \\[3mm]
         \displaystyle \frac{d v}{d z}
         &=&\displaystyle -\frac{\partial H}{\partial u}
         .
    \end{array}
    \right.
\end{gather}
The Okamoto Hamiltonian for the $i$-th \Painleve system has the form $\frac{1}{f_i(z)} h_i$, where  
\begin{align}
    \label{h6}
    \begin{aligned}
    &\begin{aligned}
    h_6
    = u^3 v^2
    - u^2 v^2
    - \kappa_1 u^2 v
    + \kappa_2 u v
    - \kappa_3 u
    + z \left(
    - u^2 v^2
    + u v^2
    + \kappa_4 u v
    + (\kappa_1 - \kappa_2 - \kappa_4) v
    \right), 
    \\[1mm]
    f_6(z)
    = z (z-1),
    \end{aligned}
    \hspace{5mm}
    \\[2mm]
    &\begin{aligned}
    h_5
    &= u^3 v^2 
    - 2 u^2 v^2 
    + v^2 
    - \kappa_1 u^2 v
    + (\kappa_1 + \kappa_2) u v
    - \kappa_2 v
    - \kappa_3 u
    + \kappa_4 z u v, 
    &&&
    f_5(z)
    &= z,
    \end{aligned}
    \\[2mm]
    &\begin{aligned}
    h_4
    &= u v^2 - u^2 v
    + \kappa_2 v
    - \kappa_3 u 
    - 2 z u v,
    &&&
    f_4(z)
    &= 1,
    \end{aligned}
    \\[2mm]
    &\begin{aligned}
    h_3^{\prime}
    &= u^2 v^2 
    + \kappa_2 u^2 v
    + \kappa_1 u v 
    + \kappa_3 u 
    + \kappa_4 z v,
    &&&
    f_3(z)
    &= z,
    \end{aligned}
    \\[2mm]
    &\begin{aligned}
    h_2
    &= - u^2 v + \tfrac12 v^2 - \kappa_3 u - \tfrac12 z v,
    &&&
    f_2(z)
    &= 1,
    \end{aligned}
    \\[2mm]
    &\begin{aligned}
    h_1
    &= - 2 u^3 + \tfrac12 v^2 - z u,
    &&&
    f_1(z)
    &= 1.
    \end{aligned}
    \end{aligned}
\end{align}
Here  $\kappa_1$, $\kappa_2$, $\kappa_3$, $\kappa_4$ are arbitrary constants.
Note that since the systems are non-autonomous, the Hamiltonians are not integrals of motion.

In \cite{Kawakami_2015} H. Kawakami constructed Hamiltonian matrix generalizations of these \Painleve $\PI-\PVI$ systems. The corresponding Hamiltonian functions have the form  $\frac{1}{f_i(z)} {\rm trace}\,(H_i),$ where $H_i$ are (non-commutative) polynomials with constant coefficients in two matrices $u$ and $v$, linear in $z$. If the size of matrices is equal to one, these Hamiltonians coincide with the Okamoto's ones.

In this paper we never use matrix entries,   but operate only with polynomials in non-commutative variables $u,v$. More rigorously, we are dealing with ODEs in free associative algebra ${\cal A}=\C[u,v]$ with the unity $\bf 1$. The independent variable $z$ plays here the role of a parameter. The corresponding definitions of trace functional and non-abelian partial derivatives used in formula \eqref{Ham} are given in Section \ref{sec:nchamsys} (see also \cite{kontsevich1993formal}).

In Section \ref{HamPen} we find all Hamiltonian non-abelian systems of \Painleve type. Our classification is based on the following assumption:
\begin{assumption}
\label{ass:ncham}
\phantom{} For each non-abelian \Painleve system \eqref{Ham} of  type $k,$ $k=1,...,6$ there exist polynomials $S^{(i)}_k\in {\cal A}$ such that
\begin{itemize}
\item[\bf 1.]  $ \frac{1}{f_k(z)}\trace S^{(1)}_k$ is the Hamiltonian of the system{\rm;} 
\item[\bf 2.] the scalar reductions of polynomials $S^{(i)}_k$ coincide with the powers $h_k^i${\rm;} 
\item[\bf 3.]   $\trace S^{(i)}_k$ and $\trace S^{(j)}_k$ commute with each other with respect to the symplectic non-abelian Poisson bracket {\rm{(}}see Section {\rm{\ref{sec:nchamsys})}} for any $i,j$.
\end{itemize}
\end{assumption}

In fact, the result of our classification coincides with  the collection of the Kawakami systems. A small generalization is the presence of the additional parameters $\beta,\gamma$ in the systems $\PIIIpr,\PV,\PVI$. 
Note that the Kawakami systems are not invariant under the simplest B\"acklund transformations (see Appendix \ref{sec:ncP6}), since these transformations change not only the parameters $\kappa_i$, but also $\beta,\gamma$.

The parameters $\beta,\gamma$ are not essential in the following sense. Matrix systems of \Painleve type with scalar coefficients are invariant under conjugations $ u \to T u T^{-1}, \, v \to T v T^{-1} $ by an arbitrary nonsingular matrix $T$. The corresponding quotient system (that is, the system satisfied by the invariants of this action) does not depend on $\beta$, $\gamma$.

In Appendix \ref{sec:syshamlist} we provide a miscellaneous information of Hamiltonian non-abelian \Painleve systems including isomonodromoc Lax representations of the form
\begin{gather} \label{eq:zerocurvcond}
    \mathbf{A}_z - \mathbf{B}_{\lambda}
    = [\mathbf{B}, \mathbf{A}]
\end{gather}
and various links between systems. 

In \cite{Balandin_Sokolov_1998, Retakh_Rubtsov_2010,  Bobrova_Sokolov_2020_1} examples of matrix Hamiltonian  $\PI$, $\PII$, and $\PIV$ systems with non-abelian (but not scalar) coefficients were found. In Section \ref{NonAb} we find $\PIIIpr$ and $\PV$ systems with one non-abelian parameter. To prove their integrability, we present isomonodromic Lax pairs 
for them. Appendix \ref{degdata_h} contains explicit formulas for the degenerations into each other of Hamiltonian systems with non-abelian parameter and their Lax pairs.

\section{Non-abelian Hamiltonian ODEs}
\label{sec:nchamsys}

\hspace{3.8mm} In this section we define the basic concepts related to non-abelian Hamiltonian systems (see \cite{kontsevich1993formal,Odesskii_Sokolov_2021}).
These systems have the form
\begin {equation} \label{geneqm}
\frac {d x _ {\alpha}} {d z} = F _ {\alpha} ({\bf x}), \qquad {\bf x} = (x_1, ..., x_N),
\end {equation}
where $ x_1, \dots, x_N $ are generators of the free associative algebra $ {\cal A} $ over $ \C $.  In fact, \eqref{geneqm}  is the notation for the derivation $d_z$ of the algebra $ {\cal A} $ such that $d_z (x_i) = F_i $. For any element $g \in {\cal A}$, the element $d_z (g)$ is uniquely determined by the Leibniz rule. Sometimes we use the notation $\frac{d}{dz}$ instead of $d_z.$
 
In the matrix case $x_i(z) \in \Mat_m (\mathbb{C})$, the scalar first integrals of   systems  \eqref{geneqm} have the form $\trace (f(x_1,...,x_N))$. Their generalization to the non-abelian case is the elements of the quotient vector space  ${\cal A} / [{\cal A}, \, {\cal A}]$.
If $a-b\in [{\cal A}, \, {\cal A}],$ we write $a\sim b.$  We denote by $\trace a$ the equivalence class of element $a \in {\cal A}$ in ${\cal A} / [{\cal A}, \, {\cal A}]$.

\begin{definition} An element $\trace \rho \in {\cal A} / [{\cal A}, \, {\cal A}]$ is called a first integral of system \eqref{geneqm} if $ d_z(\rho)~\sim~0.$
\end{definition}

In this paper we consider the case $N=2$ and denote $x_1=u,\, x_2=v.$ It is assumed that the Poisson brackets are generated by the canonical symplectic structure.  
The Hamiltonian system corresponding to the Hamiltonian $H$ has the form \eqref{Ham}, where
$H\in {\cal A}$ and $\frac{\partial}{\partial u},\, \frac{\partial}{\partial v}$ are non-abelian partial derivatives. Non-abelian derivatives of an arbitrary polynomial $f\in {\cal A}$ are defined by the identity
 $$df=\frac{\partial f}{\partial u} du+\frac{\partial f}{\partial v}dv, $$
where it is assumed that additional non-abelian symbols $du$, $dv$ are moved to the right by cyclic permutations of generators in monomials\footnote{Such an operation is equivalent to adding a commutator.}. Notice that in the non-abelian case the partial derivatives are not vector fields.

\begin{example} Let $f=u^2vuv$. We have $df=du\,uvuv+u\,du\,vuv+u^2\,dv\,uv+u^2v\,du\,v+u^2vu\, dv$. 
Performing cyclic permutations in monomials so as to move $du,~dv$ to the end of each monomial, we get   $uvuv\,du+vuvu\,du+uvu^2\,dv+vu^2v\,du+u^2vu\,dv$. Therefore,  
 $\frac{\partial f}{\partial u}=uvuv+vuvu+vu^2v,~\frac{\partial f}{\partial v}=uvu^2+u^2vu$.
 \end{example}

 \begin{remark}
 \label{rem:partder}
 It is easy to verify that $\frac{\partial}{\partial x_i} (a b- b a)=0$ for any $a,b\in  {\cal A}$ and therefore the non-abelian partial derivatives are well-defined mappings from $ {\cal A}$ to ${\cal A} / [{\cal A}, \, {\cal A}]$. The Hamiltonian of a non-abelian Hamiltonian system \eqref{Ham} is the equivalence class of the polynomial $H\in {\cal A}$ in ${\cal A} / [{\cal A}, \, {\cal A}]$ and therefore $H$ is defined up to a linear combination of commutators.
 \end{remark}
 
\begin{lemma}\label{thm:traceint} For any system of the form \eqref{Ham} the element $I=u v - v u$ is a non-abelian constant of motion{\rm:} $d_z(I)=0$.  
\end{lemma}
\begin{proof}
This statement follows from the basic identity 
\begin{equation}\label{comid}
\sum_{1\leq i\leq N}\left[x_i,\frac{\partial g}{\partial x_i}\right]=0
\end{equation}
for non-abelian partial derivatives, which is true for any $g\in {\cal A}$  \text{\rm(}see
\text{\rm\cite{kontsevich1993formal})}.
\end{proof}

\begin{corollary}\label{color1} For any system of the form \eqref{Ham} the elements ${\rm trace}(u^2 v^2-u v u v)$ and ${\rm trace}(v u^2 v u v - v u v u^2 v)$ are first integrals.
\end{corollary}
\begin{proof} It is easy to check that 
$$
u^2 v^2-u v u v \sim - \frac{1}{2} [u,\, v]^2, \qquad v u^2 v u v - v u v u^2 v \sim -\frac{1}{3} [u,\, v]^3
$$
and therefore the statement follows from Lemma \ref{thm:traceint}.
 \end{proof}

\section{Hamiltonian non-abelian \Painleve systems with constant coefficients}\label{HamPen}

\hspace{3.8mm} For a non-abelian \Painleve system of type $i$, denote by $H_i$ the non-abelian polynomial $S_i^{(1)}$ from Assumption \ref{ass:ncham} that we intend to find.
The Hamiltonian for the \Painleve system is defined by the formula
\begin{equation}\label{HAMHAM}
H_i^{ \rm H}=\frac{1}{f_i(z)} {\rm trace}\, (H_i).
\end{equation}  
By definition, $H_i$ should coincide with the polynomials $h_i$ defined by \eqref{h6} under the commutative reduction $u v = v u.$ Taking into account Remark \ref{rem:partder}, it is easy to check that the general ansatz for such non-commutative polynomials can be chosen as follows:
\begin{align}
    \label{hamPVI}
    &\begin{aligned}
    \HVI 
    = \alpha u^3 v^2 
    + (1-\alpha) u^2 v u v
    + \beta u^2 v^2
    - (1+\beta) u v u v
    - \kappa_1 u^2 v
    + \kappa_2 u v 
    - \kappa_3 u
    \hspace{1.3cm}
    \\[2mm]
    + \, z\Big(\gamma u^2 v^2
    - (1+\gamma) u v u v
    + u v^2
    + \kappa_4 u v
    + (\kappa_1-\kappa_2-\kappa_4) v\Big)
    ,
    \end{aligned}
    \\[4mm]
    \label{hamPV}
    &\begin{aligned}
    \HV = \alpha u^3 v^2 
    + (1-\alpha) u^2 v u v
    + \beta u^2 v^2
    - (2+\beta) u v u v
    + u v^2 
    - \, \kappa_1 u^2 v
    + (\kappa_1+\kappa_2) u v
    \\[2mm]
    - \kappa_2 v
    - \kappa_3 u 
    + z \kappa_4 u v
    ,
    \end{aligned}
    \\[4mm]
    \label{hamPIV}
    &\begin{aligned}
    \HIV
    = u v^2 
    - u^2 v 
    + \kappa_2 v
    -\kappa_3 u 
    - 2 z u v,
    \end{aligned}
    \\[4mm]
    \label{hamPIII'}
    &\begin{aligned}
    \HIII'
    = \beta u^2 v^2
    + (1-\beta) u v u v 
    + \kappa_2 u^2 v 
    + \kappa_1 u v
    + \kappa_3 u 
    + z \kappa_4 v,
    \end{aligned}
    \\[4mm]
    \label{hamPII}
    &\begin{aligned}
    \HII
    = - u^2 v 
    + \tfrac12 v^2 
    - \kappa_3 u 
    - \tfrac12 z v
    ,
    \end{aligned}
    \\[4mm]
    \label{hamPI}
    &\begin{aligned}
    \HI
    = - 2 u^3 
    + \tfrac12 v^2 
    - z u.
    \end{aligned}
\end{align}
 
The existence of a polynomial $S_i^{(2)}$ satisfying Assumption \ref{ass:ncham} allows one to find the unknown coefficients in these ansatzes. Due to Corollary \ref{color1} only the constant $\alpha$ needs to be defined.

\begin{proposition} 
\label{thm:S2ints}
If there exists a polynomial $S_i^{(2)}$ satisfying Assumption {\rm\ref{ass:ncham}}, then  the parameter $\alpha$ in \eqref{hamPV} and \eqref{hamPVI} is equal to zero. Other constants remain to be arbitrary.
\end{proposition}
\begin{proof}
Using as an example the case of a non-abelian system $\PIV$, we will demonstrate how to find the polynomial $S_4^{(2)}$. The Hamiltonian of the scalar system $\PIV$ is given by  
\begin{align}
    h_4
    &= u v^2 - u^2 v + \kappa_2 v - \kappa_3 u - 2 z u v,
\end{align}
and formula \eqref{hamPIV} defines its non-abelian generalization $S_4^{(1)}$.
The polynomial $h_4^2$ is equal to 
\begin{align*}
    u^4 v^2 
    - 2 u^3 v^3 
    + u^2 v^4 
    - 2 (\kappa_2 + \kappa_3)  u^2 v^2
    + 2 \kappa_2 u v^3
    + 2 \kappa_3 u^3 v
    + \kappa_2^2 v^2 
    + \kappa_3^2 u^2
    \\[1mm]
    - \, 2 u v \kappa_2 \kappa_3 
    + z \brackets{
    4 u^3 v^2 
    - 4 u^2 v^3 
    - 4 \kappa_2 u v^2
    + 4 \kappa_3 u^2 v
    }
    + 4 z^2 u^2 v^2.
\end{align*}
Let us  construct a general anzats for $S_4^{(2)}$. It is easy to verify that the sets of monomials
\begin{align*}
    \{v u^3 v u, \, v u^2 v u^2, \, v u^4 v\},
    &&
    \{v u^2 v u v, \, v u v u^2 v, \, v u v u v u,  \, v u^3 v^2\},
    &&
    \{v u v u v^2, \, v u v^2 u v, \, v u^2 v^3\}
\end{align*}
define bases in the vector subspaces of homogeneous polynomials of degrees (4,2),\,(3,3) and (2,4)  in $u$ and $v$  projected onto the quotient space ${\cal A} / [{\cal A }, \, {\cal A}]$. 
 Therefore, the leading part of the ansatz can be taking as follows:
$$\begin{array}{r}
    a_1 v u^3 v u 
    + a_2 v u^2 v u^2 
    + (1 - a_1 - a_2) v u^4 v 
    + b_1 v u^2 v u v 
    + b_2 v u v u^2 v 
    + b_3 v u v u v u 
    \\[3mm]
    + \, (
    - 2 - b_1 - b_2 - b_3
    ) v u^3 v^2
    + c_1 v u v u v^2 
    + c_2 v u v^2 u v
    + (
    1 - c_1 - c_2
    ) v u^2 v^3.
    \end{array}
$$
A similar consideration for lower terms leads to 
\begin{align*}
    S_4^{(2)}
    = a_1 v u^3 v u 
    + a_2 v u^2 v u^2 
    + (1 - a_1 - a_2) v u^4 v 
    + b_1 v u^2 v u v 
    + b_2 v u v u^2 v 
    + b_3 v u v u v u 
    \\ 
    + \, (
    - 2 - b_1 - b_2 - b_3
    ) v u^3 v^2
    + c_1 v u v u v^2 
    + c_2 v u v^2 u v
    + (
    1 - c_1 - c_2
    ) v u^2 v^3
    \\
    + \, d_1 v u v u 
    + (-2 (\kappa_2 + \kappa_3) - d_1) v u^2 v 
    + 2 \kappa_2 v u v^2 
    + 2 \kappa_3 v u^3
    + \kappa_2^2 v^2 
    + \kappa_3^2 u^2
    \\
    - \, 2 \kappa_2 \kappa_3 v u 
    + z \left(
    e_1 v u^2 v u 
    + (4 - e_1) v u^3 v 
    + f_1 v u v u v
    + (-4 - f_1) v u^2 v^2 
    \right.
    \\
    \left.
    - \, 4 \kappa_2 v u v 
    + 4 \kappa_3 v u^2
    \right)
    + z^2 \brackets{
    g_1 v u v u
    + (4 - g_1) v u^2 v 
    }.
\end{align*}
The condition\footnote{Here the independent variable $z$ plays the role of a parameter.}
\begin{equation}
   \trace \Big(\frac{\partial S_4^{(2)}}{\partial u} \,\, \frac{\partial S_4^{(1)}}{\partial v}
    - \frac{\partial S_4^{(2)}}{\partial v} \,\, \frac{\partial S_4^{(1)}}{\partial u}\Big)
    = 0
\end{equation}
is equivalent to an overdetermined algebraic system for the coefficients of $H_4$ and $S_4^{(2)}$. Solving it, we obtain
\begin{align*}
    a_1
    &= 0, 
    &
    a_2
    &= 1,
    &
    b_2 
    &= - 2 - b_1,
    &
    b_3
    &= 0,
    &
    c_1
    &= 0,
    &
    c_2
    &= 1,
    &
    e_1
    &= 4,
    &
    f_4
    &= - 4.
\end{align*}
Therefore, the polynomial $S_4^{(2)}$ has the form
\begin{align*}
    S_4^{(2)}
    = v u^2 v u^2 
    - 2 v u v u^2 v 
    + v u v^2 u v
    - 2 (\kappa_2 + \kappa_3) v u^2 v 
    + 2 \kappa_2 v u v^2 
    + 2 \kappa_3 v u^3
    + (\kappa_2 v - \kappa_3 u)^2
    \\
    + \, z \left(
    4 v u^2 v u 
    - 4 v u v u v
    - \, 4 \kappa_2 v u v 
    + 4 \kappa_3 v u^2
    \right)
    + 4 z^2 v u^2 v
    + b_1 \brackets{
    v u^2 v u v - v u v u^2 v
    }
    \\
    - \, (d_1 + g_1 z^2) \brackets{
    v u^2 v - v u v u
    }
    .
\end{align*}
The presence of the constants $b_1, d_1$ and $g_1$ is explained by Corollary \ref{color1}. 

Similar calculations in the cases of $\PV$ and $\PVI$ are more laborious. The algebraic system contains the parameter $\alpha$ and it's compatibility conditions give rise to the equality $\alpha=0$. 
\end{proof}
 
In Appendix \ref{sec:syshamlist} we present isomonodromic Lax representations for the equations of motions defined by the non-abelian Hamiltonians and explicit expressions for limiting transitions connecting Hamiltonian systems. 

\section{Systems with non-abelian constants}\label{NonAb}
\hspace{3.8mm}

In the paper \cite{Balandin_Sokolov_1998} the matrix Painlev\'e-1 equation with an arbitrary constant matrix (or non-abelian constant) $h$ arose. It can be written as the system\footnote{Formally speaking, scalar constants like $z$ or $\kappa_i$ in the formulas below should be replaced by $z\,{\bf 1},$ $\kappa_i\, {\bf 1}$,~etc.}
\begin{equation}
    \left\{
    \begin{array}{lcl}
         u'
         &=& v,
         \\[2mm]
         v'
         &=& 6 u^2 + z + h.
    \end{array}
    \right.
    \end{equation}
A similar Painlev\'e-2 system,
\begin{equation}
    \left\{
    \begin{array}{lcl}
         u'
         &=& - u^2 + v - \tfrac12 z - h,
         \\[2mm]
         v'
         &=& v u + u v + \kappa_3,
    \end{array}
    \right.
    \end{equation}
can be extracted  from \cite{Retakh_Rubtsov_2010}\footnote{To do this, one has to replace the non-commutative independent variable $\tilde z$ by $z\, \mathbf{1} + 2 h$, where $h$ is an arbitrary non-abelian constant.}, while a system of Painlev\'e-4 type
\begin{equation}
    \label{eq:ncsysP4_nc}
    \left\{
    \begin{array}{lcl}
     u' 
     &=& 
     - u^2 
     + u v + v u
     - 2 z u
     + h u
     + \kappa_2,
     \\[2mm]
     v'
     &=& 
     - v^2
     + v u + u v
     + 2 z v
     - v h
     + \kappa_3,
    \end{array}
    \right.
\end{equation}
with an arbitrary constant matrix $h$ was found in  \cite{Bobrova_Sokolov_2020_1}. All these systems are Hamiltonian in the sense of Section \ref{sec:nchamsys}. The Hamiltonians are given by the formulas
\begin{align}
    \label{hamPIV_nc}
    &\begin{aligned}
    \HIV^{\rm{H}}
    = u v^2 
    - u^2 v 
    + \kappa_2 v
    -\kappa_3 u 
    - 2 z u v
    + h \, u v
    ,
    \end{aligned}
    \\[4mm]
    \label{hamPII_nc}
    &\begin{aligned}
    \HII^{\rm{H}}
    = - u^2 v 
    + \tfrac12 v^2 
    - \kappa_3 u 
    - \tfrac12 z v
    - h \, v
    ,
    \end{aligned}
    \\[4mm]
    \label{hamPI_nc}
    &\begin{aligned}
    \HI^{\rm{H}}
    = - 2 u^3 
    + \tfrac12 v^2 
    - z u
    - h \, u
    ,
    \end{aligned}
\end{align}
where $h \in \mathcal{A}$. Note that after the transition to scalar variables in these systems, the constant $h$ can be reduced to zero by a shift $z$. In the non-abelian case such a shift is impossible because $z$ is a scalar variable. It is a non-trivial fact that the above three systems possess isomonodromic Lax pairs. There are non-abelian non-Hamiltonian Painlev\'e systems of $\PII$ and $\PIV$ type \cite{Adler_Sokolov_2020,Bobrova_Sokolov_2020_1} with non-abelian coefficients that cannot be constructed using this simple trick.

In this section we present non-abelian Hamiltonian Painlev\'e systems of $\PIIIpr$ and $\PV$ type with one arbitrary non-abelian coefficient. In the scalar systems corresponding to them, the parameter $h$ can be reduced to one by a scaling of the variable $z$.
 
For the constructed non-abelian systems we were able to find isomonodromic Lax representations only for special values of the parameter $\beta$. This is not surprising, since the arbitrariness of the parameters $\beta, \gamma$ in non-abelian Hamiltonians is related to their gauge invariance under conjugations, which is absent in the case of systems with non-abelian coefficients.
 
The scalar \PPainleve-$3^{\prime}$ system depends on two essential parameters while other parameters can be normalized by scallings. In particular, 
in the case $\kappa_4 \neq 0$ we can reduce $\kappa_4$ to 1 by using a scaling of $z$.  After the formal replacement of the scalar parameter $\kappa_4$ by $h$, the polynomial \eqref{hamPIII'} becomes
\begin{align}
    H_3^\prime
    &= \beta u^2 v^2 + (1 - \beta) u v u v
    + \kappa_2 u^2 v
    + \kappa_1 u v 
    + \kappa_3 u
    + z \, h v.
\end{align}
The corresponding Hamiltonian system has the form
\begin{gather}
    \label{eq:ncsysP3_nc}
    \left\{
    \begin{array}{lcl}
         z \, u'
         &=& 2 u v u
         + \beta [u, [u, v]]
         + \kappa_1 u
         + \kappa_2 u^2 + z h,
         \\[2mm]
         z \, v'
         &=& - 2 v u v 
         + \beta [v, [u, v]]
         - \kappa_1 v
         - \kappa_2 u v - \kappa_2 v u - \kappa_3.
    \end{array}
    \right.
\end{gather}
If $\beta=0$ this system has the isomonodromic Lax pair with 
\begin{align}
    \mathbf{A} (\lambda, z)
    &= 
    \begin{pmatrix}
    0 & 0 \\[0.9mm]
    0 & - z \, h
    \end{pmatrix}
    -
    \begin{pmatrix}
    \kappa_1
    & 
    u
    \\[0.9mm]
    u v^2
    + \kappa_1 v
    + \kappa_2 u v 
    + \kappa_3
    & - u v + v u
    \end{pmatrix}
    \lambda^{-1} 
    +
    \begin{pmatrix}
    v + \kappa_2
    & 
    - 1
    \\[0.9mm]
    v^2 + \kappa_2 v
    &
    - v
    \end{pmatrix}
    \lambda^{-2},
    \\
    \mathbf{B} (\lambda, z)
    &= 
    \begin{pmatrix}
    0 & 0 
    \\[0.9mm]
    0 & -h
    \end{pmatrix}
    \lambda 
    + z^{-1}
    \begin{pmatrix}
    u v + v u
    + \kappa_2 u 
    &
    - u
    \\[0.9mm]
    - \brackets{
    u v^2 
    + \kappa_1 v
    + \kappa_2 u v + \kappa_3
    }
    & 
    0
    \end{pmatrix}
    + z^{-1} [u, v] \, \mathbf{I}
    .
\end{align}

The same trick produces a $\PV$-type Painlev\'e system with the arbitrary non-abelian constant $h$. Replacing the parameter $\kappa_4$  in \eqref{hamPV} with $h$, we get the polynomial
\begin{gather} 
    H_5
    = u^2 v u v 
    + \beta u^2 v^2 - (2 + \beta) u v u v
    + u v^2
    - \kappa_1 u^2 v
    + (\kappa_1 + \kappa_2) u v
    - \kappa_2 v
    - \kappa_3 u
    + z h \, u v, 
\end{gather}
which defines the following Hamiltonian system
\begin{gather}
    \label{eq:ncsysP5_nc}
    \left\{
    \begin{array}{lcr}
         z \, u'
         &=& u^2 v u + u v u^2
         - 4 u v u
         + \beta [u, [u, v]]
         - \kappa_1 u^2 + u v + v u 
         + (\kappa_1 + \kappa_2) u
         \\[1mm]
         && 
         - \, \kappa_2
         + z \, h u
         ,
         \\[2mm]
         z \, v'
         &=& - u v u v - v u^2 v - v u v u
         + 4 v u v 
         + \beta [v, [u, v]]
         + \kappa_1 u v + \kappa_1 v u - v^2
         \hspace{0.4cm}
         \\[1mm]
         && 
         - \, (\kappa_1 + \kappa_2) v
         + \kappa_3
         - z  \, v h
         .
    \end{array}
    \right.
\end{gather}
One can check that  
\begin{align}
    \mathbf{A} (\lambda, z)
    &= 
    \begin{pmatrix}
    0 & 0 \\[0.9mm] 0 & - z \, h
    \end{pmatrix}
    + 
    \begin{pmatrix}
    - u v + \kappa_1 & 1
    \\[0.9mm]
    - u v u v 
    + \kappa_1 u v + \kappa_3
    & u v
    \end{pmatrix}
    \lambda^{-1}
    + 
    \begin{pmatrix}
    u v - \kappa_2 & - u 
    \\[0.9mm]
    v u v - \kappa_2 v & - v u
    \end{pmatrix}
    (\lambda - 1)^{-1}
    ,
    \\[-4mm]
    \mathbf{B} (\lambda, z)
    &= 
    \begin{pmatrix}
    0 & 0 \\[0.9mm] 0 & - h
    \end{pmatrix}
    \lambda
    + z^{-1}
    \begin{pmatrix}
    u^2 v + u v u
    - 2 u v 
    - \kappa_1 u + \kappa_1
    & 
    - u + 1
    \\[0.9mm]
    - u v u v + v u v 
    + \kappa_1 u v 
    - \kappa_2 v + \kappa_3
    &
    0
    \end{pmatrix}
    + z^{-1} [u, v] \, \mathbf{I}
    + h \ \mathbf{I}
\end{align}
defines an isomonodromic Lax pair for this system in the case $\beta=-1$.

One more Hamiltonian system with a non-abelian parameter corresponds to a non-polynomial degeneration $\PIIIpr (D_7)$ of a non-abelian  system \Painleve $\PIIIpr$. The Hamiltonian of this system is given by the formula
\begin{align}
    \label{hamPIII'D7_nc}
    z H_3^{\prime {\rm H}}(D_7)
    &= \beta u^2 v^2 + (1 - \beta) u v u v
    + \kappa_2 u^2 v
    + \kappa_1 u v
    + \kappa_3u
    + z \, h u^{-1}.
\end{align}
The corresponding Hamiltonian system
\begin{gather}
    \label{eq:ncsysP3D7_nc}
    \left\{
    \begin{array}{lcl}
         z \, u'
         &=& 2 u v u
         + \beta [u, [u, v]]
         + \kappa_1 u
         + \kappa_2 u^2
         ,
         \\[2mm]
         z \, v'
         &=& - 2 v u v
         + \beta [v, [u, v]]
         -\kappa_1 v
         - \kappa_2 u v - \kappa_2 v u - \kappa_3
         + z \, u^{-1} h u^{-1}
         ,
    \end{array}
    \right.
\end{gather}
in the case $\beta = 0$ has a Lax pair of the form \eqref{eq:matABform_P3D7'_} with matrix coefficients
\begin{gather}
    \notag
    \begin{aligned}
    A_0
    &=
    \begin{pmatrix}
    0 & 0 \\[0.9mm]
    z \, h u^{-1} & 0
    \end{pmatrix},
    &
    A_{-1}
    &=-
    \begin{pmatrix}
    u v + \kappa_1
    &
    u
    \\[0.9mm]\kappa_2 u v + \kappa_3
    & - u v
    \end{pmatrix},
    &
    A_{-2}
    &=
    \begin{pmatrix}
    -\kappa_2
    &
    1
    \\[0.9mm]
    0
    &
    0
    \end{pmatrix},
    \end{aligned}
    \\[2mm]
    \notag
    \begin{aligned}
    B_1
    &=
    \begin{pmatrix}
    0 & 0
    \\[0.9mm]
    h u^{-1} & 0
    \end{pmatrix}
    ,
    &&&
    B_0
    &= z^{-1}
    \begin{pmatrix}
    u v
    + \kappa_2 u
    &
    -u
    \\[0.9mm]
    0
    &
    u v
    \end{pmatrix}
    .
    \end{aligned}
\end{gather}

 The considered systems  with the non-abelian coefficient $h$  and their Lax pairs are related to each other by a limiting transitions given in Appendix \ref{degdata_h}.
 
\subsubsection*{Acknowledgements}

The authors are grateful to M.~A.~Bershtein and B.~I.~Suleimanov for useful discussions. They are thankful to IHES for hospitality. The research of the second author was carried out under the State Assignment 0029-2021-0004 (Quantum field theory) of the Ministry of Science and Higher Education of the Russian Federation. The first author was partially supported by the International Laboratory of Cluster Geometry HSE, RF Government grant № 075-15-2021-608, and by Young Russian Mathematics award.

\appendix
\section*{Appendices}

\section{Non-abelian Hamiltonian systems of \PPainleve-type with constant coefficients}
\label{sec:syshamlist}

\qquad In this section explicit formulas are given for the degenerations of non-abelian Hamiltonian systems described by the scheme
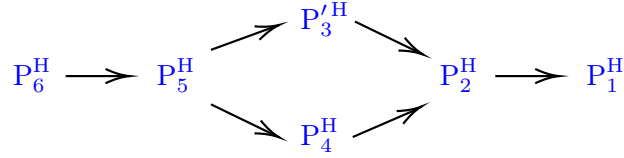
\begin{figure}[H]
     \centering
     \scalebox{1.1}{\tikzset{every picture/.style={line width=0.75pt}} 

\begin{tikzpicture}[x=0.75pt,y=0.75pt,yscale=-1,xscale=1]

\draw    (32.65,36.5) -- (59.88,36.5) ;
\draw [shift={(61.88,36.5)}, rotate = 180] [color={rgb, 255:red, 0; green, 0; blue, 0 }  ][line width=0.75]    (10.93,-3.29) .. controls (6.95,-1.4) and (3.31,-0.3) .. (0,0) .. controls (3.31,0.3) and (6.95,1.4) .. (10.93,3.29)   ;
\draw    (99.3,24.46) -- (128.8,13.68) ;
\draw [shift={(130.67,12.99)}, rotate = 159.93] [color={rgb, 255:red, 0; green, 0; blue, 0 }  ][line width=0.75]    (10.93,-3.29) .. controls (6.95,-1.4) and (3.31,-0.3) .. (0,0) .. controls (3.31,0.3) and (6.95,1.4) .. (10.93,3.29)   ;
\draw    (164.3,64.46) -- (193.8,51.68) ;
\draw [shift={(195.67,50.5)}, rotate = 152.93] [color={rgb, 255:red, 0; green, 0; blue, 0 }  ][line width=0.75]    (10.93,-3.29) .. controls (6.95,-1.4) and (3.31,-0.3) .. (0,0) .. controls (3.31,0.3) and (6.95,1.4) .. (10.93,3.29)   ;
\draw    (99.3,49.46) -- (128.72,63.72) ;
\draw [shift={(130.52,64.59)}, rotate = 205.86] [color={rgb, 255:red, 0; green, 0; blue, 0 }  ][line width=0.75]    (10.93,-3.29) .. controls (6.95,-1.4) and (3.31,-0.3) .. (0,0) .. controls (3.31,0.3) and (6.95,1.4) .. (10.93,3.29)   ;
\draw    (165.3,11.46) -- (194.72,25.72) ;
\draw [shift={(196.52,26.59)}, rotate = 205.86] [color={rgb, 255:red, 0; green, 0; blue, 0 }  ][line width=0.75]    (10.93,-3.29) .. controls (6.95,-1.4) and (3.31,-0.3) .. (0,0) .. controls (3.31,0.3) and (6.95,1.4) .. (10.93,3.29)   ;
\draw    (229.65,36.5) -- (256.88,36.5) ;
\draw [shift={(258.88,36.5)}, rotate = 180] [color={rgb, 255:red, 0; green, 0; blue, 0 }  ][line width=0.75]    (10.93,-3.29) .. controls (6.95,-1.4) and (3.31,-0.3) .. (0,0) .. controls (3.31,0.3) and (6.95,1.4) .. (10.93,3.29)   ;

\draw (7.02,26.8) node [anchor=north west][inner sep=0.75pt]    
{\ref{eq:ncsysP6}};
\draw (73.02,26.8) node [anchor=north west][inner sep=0.75pt]    
{\ref{eq:ncsysP5}};
\draw (138.8,1.04) node [anchor=north west][inner sep=0.75pt]   
{\ref{eq:ncsysP3'}};
\draw (138.8,54.4) node [anchor=north west][inner sep=0.75pt]    
{\ref{eq:ncsysP4}};
\draw (202.77,26.8) node [anchor=north west][inner sep=0.75pt]    
{\ref{eq:ncsysP2}};
\draw (270.02,26.8) node [anchor=north west][inner sep=0.75pt]    
{\ref{eq:ncsysP1}};

\end{tikzpicture}}
\end{figure}

As in the scalar case, these degenerations are given by formulas of the form
\begin{gather}
    \begin{aligned}
    z
    &= Z (\varepsilon, \tilde z),
    &&&&&
    u
    &= U (\varepsilon, \tilde u, \tilde v),
    &&&&&
    v
    &= V (\varepsilon, \tilde u, \tilde v),
    \end{aligned}
    \\[2mm]
    \begin{aligned}
        \kappa_i
        &= K_i(
        \varepsilon,
        \tilde \kappa_1, \tilde \kappa_2, \tilde \kappa_3, \tilde \kappa_4
        ),
        &&&&
        i
        &= 1, 2, 3, 4,
    \end{aligned}
\end{gather}
where $Z$, $U$, $V$, and $K_i$ are linear functions in all their arguments except $\varepsilon$.
After passing to the limit $\varepsilon \to 0$, we obtain a Painlev\'e system with a smaller number with respect to the variables $\tilde z$, $\tilde u$, $\tilde v$ with parameters $\tilde \kappa_j$. Hoping that this will not lead to misunderstandings, we omit the sign $\tilde{\phantom{u}}$ everywhere in the   system obtained as a result of the limiting transition.

To degenerate Lax pairs, the transformation 
\begin{align}
    \lambda
    &= \Lambda (\varepsilon, \tilde \lambda)
\end{align}
  is applied, where $\Lambda$ is a linear in $\tilde \lambda$. Furthermore, a gauge transformation
\begin{align}
    &&
    {\mathbf{A}} (\lambda, z)
    &= g^{-1} \, \tilde {\mathbf{A}} (\lambda, z) \, g
    - g^{-1} \, g_{\lambda}^{\prime} ,
    &
    {\mathbf{B}} (\lambda, z)
    &= g^{-1} \, \tilde {\mathbf{B}} (\lambda, z) \, g
    - g^{-1} \, g_{z}^{\prime}
    &&
\end{align}
with the help of some nondegenerate matrix $g = g(\varepsilon, \lambda, z)$ is used. When we transform the matrices $\mathbf{A} (\lambda, z)$, $\mathbf{B} (\lambda, z)$, the variables $(z, \lambda, u, v)$ and the parameters $\kappa_i $, we mean that the transformation of the variables follows the gauge transformation. In the resulting Lax pair, we will also omit the  sign $\tilde{\phantom{u}}$.

\subsection{\PPainleve-6 systems}
\label{sec:ncP6}
\hspace{3.8mm}
Hamiltonian systems corresponding to the polynomial $H_6$ \eqref{hamPVI} with $\alpha=0$ have the form
\begin{gather}
    \label{eq:ncsysP6}
    \tag*{$\PVIn{\rm{H}}$}
    \left\{
    \begin{array}{lcr}
         z (z - 1) \, u'
         &=& u^2 v u 
         + u v u^2 
         - 2 u v u
         + \beta [u, [u, v]]
         - \kappa_1 u^2 
         + \kappa_2 u 
         \hspace{2.4cm}
         \\[1mm]
         && 
         + \, z \brackets{
         - 2 u v u 
         + \gamma [u, [u, v]]
         + u v + v u
         + \kappa_4 u 
         + (\kappa_1 - \kappa_2 - \kappa_4)
         },
         \\[2mm]
         z (z - 1) \, v'
         &=& - u v u v - v u^2 v - v u v u
         + 2 v u v
         + \beta [v, [u, v]]
         + \kappa_1 (u v +  v u)
         - \kappa_2 v
         \\[1mm]
         && 
         + \, \kappa_3
         + z \brackets{
         2 v u v + \gamma [v, [u, v]]
         - v^2
         - \kappa_4 v
         }.
    \end{array}
    \right.
\end{gather}
This family of systems is invariant under the group $S_3$ generated by the transformations
\begin{gather}
    \label{eq:scalP6sym1}
    \begin{aligned}
    r_1
    \brackets{
    z, u, v
    }
    &= \brackets{
    1 - z, \, 
    1 - u, \,
    - v
    },
    &&&&&
    r_2
    \brackets{
    z, u, v
    }
    &= \brackets{
    z^{-1}, \,
    z^{-1} u, \,
    z v
    }.
    \end{aligned}
\end{gather}
The involutions $r_1$ and $r_2$ change the parameters as follows:
\begin{align}
    \label{eq:scalP6sym1k}
    &\begin{aligned}
    r_1
    \brackets{
    \kappa_1, 
    \kappa_2, 
    \kappa_3, 
    \kappa_4,
    \beta,
    \gamma
    }
    &= \brackets{
    \kappa_1, \,
    2 \kappa_1 - \kappa_2 - \kappa_4, \,
    \kappa_3, \,
    \kappa_4, \,
    - 1 - \beta - \gamma, \,
    \gamma
    },
    \end{aligned}
    \\[2mm]
    \label{eq:scalP6sym2k}
    &\begin{aligned}
    r_2
    \brackets{
    \kappa_1, 
    \kappa_2, 
    \kappa_3, 
    \kappa_4,
    \beta,
    \gamma
    }
    &= \brackets{
    \kappa_1, \,
    \kappa_4 - 1, \,
    \kappa_3, \,
    \kappa_2 + 1,
    \gamma, \,
    \beta
    }.
    \end{aligned}
    &&&&&&&&
\end{align}

System \ref{eq:ncsysP6} is equivalent to the zero-curvature condition \eqref{eq:zerocurvcond} with
\begin{align} \label{eq:matABform}
    \mathbf{A} (z, \lambda)
    &= \dfrac{A_0}{\lambda}
    + \dfrac{A_1}{\lambda - 1}
    + \dfrac{A_2}{\lambda - z},
    &
    \mathbf{B} (z, \lambda)
    &= - \dfrac{A_2}{\lambda - z}
    + B_{1},
\end{align}
where
\begin{gather} 
    \notag
    \begin{aligned}
    A_0
    &= 
    \begin{pmatrix}
    - 1
    - \kappa_1 
    + \kappa_4
    & 
    u z^{-1} - 1
    \\[0.9mm]
    0 & 0
    \end{pmatrix},
    &&&
    A_1
    &= 
    \begin{pmatrix}
    - u v + \kappa_1
    & 
    1 
    \\[0.9mm]
    - u v u v + \kappa_1 u v + \kappa_3
    & 
    u v
    \end{pmatrix},
    \end{aligned}
    \\[3mm]
    \label{eq:Laxpair_P6case1}
    A_2
    = 
    \begin{pmatrix}
    u v 
    + (\kappa_1 - \kappa_2 - \kappa_4) 
    & 
    - u z^{-1}
    \\[0.9mm]
    z v u v
    + (\kappa_1 - \kappa_2 - \kappa_4)  z v 
    & 
    - v u
    \end{pmatrix},
    \\[3mm]
    \notag
    \begin{aligned}
    B_{1}
    = 
    \begin{pmatrix}
    (z (z - 1))^{-1} \brackets{
    u^2 v + u v u
    - \kappa_1 u 
    - z \brackets{
    u v + v u + (\kappa_1 - \kappa_2 - \kappa_4)
    }
    }
    & 0
    \\[0.9mm]
    - v u v 
    - (\kappa_1 - \kappa_2 - \kappa_4)  v
    & 0
    \end{pmatrix}
    \\[2mm]
    - \, (z (z - 1))^{-1} \brackets{
    \beta + (1 + \gamma) z
    } [u, v] \, \mathbf{I}
    .
    \end{aligned}
\end{gather}

\subsection{\PPainleve-5 systems}
\hspace{3.8mm}
The Hamiltonian system generated by the polynomial \eqref{hamPV} with $\alpha = 0$ is written as
\begin{gather}
    \label{eq:ncsysP5}
    \tag*{$\PVn{\rm{H}}$}
    \left\{
    \begin{array}{lcr}
         z \, u'
         &=& u^2 v u + u v u^2
         - 4 u v u
         + \beta [u, [u, v]]
         - \kappa_1 u^2 + u v + v u 
         + (\kappa_1 + \kappa_2) u
         \\[1mm]
         && 
         - \, \kappa_2
         + \, \kappa_4 z u
         ,
         \\[2mm]
         z \, v'
         &=& - u v u v - v u^2 v - v u v u
         + 4 v u v 
         + \beta [v, [u, v]]
         + \kappa_1 u v + \kappa_1 v u - v^2
         \hspace{0.4cm}
         \\[1mm]
         && 
         - \, (\kappa_1 + \kappa_2) v
         + \kappa_3
         - \, \kappa_4 z v
         .
    \end{array}
    \right.
\end{gather}
To get system \ref{eq:ncsysP5} from system \ref{eq:ncsysP6}, one can
make the following substitution of variables and parameters:
\begin{align} \label{eq:P6toP5map}
    z 
    &= 1 + \varepsilon \, \tilde z,
    &
    \kappa_2 
    &= - \varepsilon^{-1} \tilde \kappa_4 
    + \tilde \kappa_2,
    &
    \kappa_4
    &= \varepsilon^{-1} \tilde \kappa_4 
    + \tilde \kappa_1,
    &
    \beta 
    &= \tilde \beta - \gamma,
\end{align}
and pass to the limit $\varepsilon \to 0$. In order to construct a Lax pair for the system \ref{eq:ncsysP5}, we supplement \eqref{eq:P6toP5map} with the formula
\begin{align} \label{eq:P6toP5map_la}
    \lambda 
    &= 1 + \varepsilon \, \tilde z \, \tilde \lambda.
\end{align}
As a result of  the limiting transition, the Lax pair \eqref{eq:matABform} passes to the pair
\begin{align} \label{eq:matABform_P5}
    \mathbf{A} (\lambda, z)
    &= A_0
    + \frac{A_1}{\lambda}
    + \frac{A_2}{\lambda - 1},
    &
    \mathbf{B} (\lambda, z)
    &= B_1 \lambda
    + B_{0},
\end{align}
where matrices $A_0$, $A_1$, $A_2$, $B_1$, $B_0$ are given by
\begin{gather} 
    \notag
    \begin{aligned}
    A_{0}
    &= 
    \begin{pmatrix}
    \kappa_4 z & 0 \\[0.9mm] 0 & 0
    \end{pmatrix},
    &
    A_1
    &= 
    \begin{pmatrix}
    - u v + \kappa_1 & 1
    \\[0.9mm]
    - u v u v 
    + \kappa_1 u v + \kappa_3
    & u v
    \end{pmatrix},
    &
    A_2
    &= 
    \begin{pmatrix}
    u v - \kappa_2 & - u 
    \\[0.9mm]
    v u v - \kappa_2 v & - v u
    \end{pmatrix}
    ,
    \end{aligned}
    \\[-1mm]
    \label{eq:Laxpair_P5case1_bet}
    \\[-1mm]
    \notag
    \begin{aligned}
    B_1
    &= 
    \begin{pmatrix}
    \kappa_4 & 0 \\[0.9mm] 0 & 0
    \end{pmatrix},
    &
    B_{0}
    &= z^{-1}
    \begin{pmatrix}
    u^2 v + u v u
    - 2 u v 
    - \kappa_1 u + \kappa_1
    & 
    - u + 1
    \\[0.9mm]
    - u v u v + v u v 
    + \kappa_1 u v 
    - \kappa_2 v + \kappa_3
    &
    0
    \end{pmatrix}
    - z^{-1} \beta [u, v] \, \mathbf{I}
    .
    \end{aligned}
\end{gather}

\subsection{The \PPainleve-4 system}
\hspace{3.8mm}
The Hamiltonian $H_4^{\rm{H}}$ with \eqref{hamPIV} leads to the system
\begin{gather}
    \tag*{$\PIVn{\rm{H}}$}
    \label{eq:ncsysP4}
    \left\{
    \begin{array}{lcl}
         u' 
         &=& 
         - u^2 
         + u v + v u
         - 2 z u
         + \kappa_2,
         \\[2mm]
         v'
         &=& 
         - v^2
         + v u + u v
         + 2 z v
         + \kappa_3.
    \end{array}
    \right.
\end{gather}
Making the following replacement with the small parameter $\varepsilon$:
\begin{gather}
    \label{eq:P5toP4map}
    \begin{aligned}
    z 
    &= 1 + \sqrt{2} \varepsilon \, \tilde z
    ,
    &&&
    u 
    &= \tfrac{1}{\sqrt{2}} \varepsilon \, \tilde u,
    &&&
    v
    &= \tfrac{1}{\sqrt{2}} \varepsilon^{-1} \, \tilde v,
    \end{aligned}
    \\[2mm]
    \label{eq:P5toP4map_k}
    \begin{aligned}
    \kappa_1
    &= \varepsilon^{-2},
    &&&
    \kappa_2 
    &= - \tfrac12 \tilde \kappa_2,
    &&&
    \kappa_3
    &= \tfrac12 \varepsilon^{-2} \tilde \kappa_3,
    &&&
    \kappa_4
    &= - \varepsilon^{-2},
    \end{aligned}
\end{gather}
and passing to the limit,  we transform  system \ref{eq:ncsysP5} to  system \ref{eq:ncsysP4}. To extend this limiting transition to Lax pairs, consider the gauge transformation
\begin{align}
    \label{eq:P5toP4map_la}
    &&
    \lambda 
    &= 1 + \sqrt{2} \varepsilon \, \tilde \lambda
    ,
    &
    g
    &= 
    \begin{pmatrix}
    1 & 0 \\ 0 & \sqrt{2} \, \varepsilon
    \end{pmatrix}.
    &&
\end{align}
One can check that it degenerates the pair \eqref{eq:matABform_P5} into the pair
\begin{align} \label{eq:matABform_P4}
    \mathbf{A} (\lambda, z)
    &= A_1 \lambda + A_{0} + A_{-1} \lambda^{-1},
    &
    \mathbf{B} (\lambda, z)
    &= B_1 \lambda + B_{0},
\end{align}
with
\begin{equation} \label{eq:Laxpair_P40}
    \begin{gathered}
        \begin{aligned}
        A_1
        &= 
        \begin{pmatrix}
        - 2 & 0 \\[0.9mm] 0 & 0
        \end{pmatrix},
        &&&
        A_0
        &= 
        \begin{pmatrix}
        - 2 z & 1 \\[0.9mm] 
        u v + \kappa_3 & 0
        \end{pmatrix},
        &&&
        A_{-1}    
        &= \tfrac12
        \begin{pmatrix}
        u v + \kappa_2 & - u \\[0.9mm]
        v u v + \kappa_2 v & - v u
        \end{pmatrix},
        \end{aligned}
        \\[2mm]
        \begin{aligned}
        B_1
        &= 
        \begin{pmatrix}
        - 2 & 0 \\[0.9mm] 0 & 0
        \end{pmatrix},
        &&&
        B_0
        &= 
        \begin{pmatrix}
        - u - 2 z & 1 \\[0.9mm]
        u v + \kappa_3 & 0
        \end{pmatrix}.
        \end{aligned}
    \end{gathered}
\end{equation}

\subsection{\PPainleve-\texorpdfstring{$\mathbf{3^\prime}$}{3'} systems}
\hspace{3.8mm}
The Hamiltonian system corresponding to \eqref{hamPIII'} has the form
\begin{gather}
    \tag*{$\PIIIprn{\rm{H}}$}
    \label{eq:ncsysP3'}
    \left\{
    \begin{array}{lcl}
         z \, u'
         &=& 2 u v u 
         + \beta [u, [u, v]]
         + \kappa_1 u
         + \kappa_2 u^2 + \kappa_4 z,
         \\[2mm]
         z \, v'
         &=& - 2 v u v 
         + \beta [v, [u, v]]
         - \kappa_1 v
         - \kappa_2 u v - \kappa_2 v u - \kappa_3
         .
    \end{array}
    \right.
\end{gather}
The following limiting transition:
\begin{gather}
\label{eq:P5toP3D6'map}
    \begin{aligned}
    u 
    &= 1 + \varepsilon \, \tilde u,
    &&&
    v
    &= \varepsilon^{-1} \, \tilde v,
    \end{aligned}
    \\[2mm]
    \label{eq:P5toP3D6'map_k}
    \begin{aligned}
    \kappa_1
    &= - \tilde \kappa_1 
    - \varepsilon^{-1} \tilde \kappa_2,
    &
    \kappa_2 
    &= - \varepsilon^{-1} \tilde \kappa_2,
    &
    \kappa_3
    &= - \varepsilon^{-1} \tilde \kappa_3,
    &
    \kappa_4
    &= \varepsilon \, \tilde \kappa_4,
    &
    \beta
    &= \tilde \beta - 1,
    \end{aligned}
\end{gather}
implements the degeneration of $\text{\ref{eq:ncsysP5}} \to \text{\ref{eq:ncsysP3'}}$. To degenerate the Lax pair, one can use the gauge transformation
\begin{align}
    \label{eq:P5toP3D6'map_la}
    &&
    \lambda 
    &= \varepsilon^{-1} \tilde \lambda,
    &
    g
    &= 
    \begin{pmatrix}
    1 & 0 \\[0.9mm] 0 & \varepsilon
    \end{pmatrix},
    &&
\end{align}
reducing pair \eqref{eq:matABform_P5} to
\begin{align} \label{eq:matABform_P3D6'_}
    \mathbf{A} (\lambda, z)
    &= A_0 
    + A_{-1} \lambda^{-1} 
    + A_{-2} \lambda^{-2},
    &
    \mathbf{B} (\lambda, z)
    &= B_1 \lambda + B_0,
\end{align}
where 
\begin{gather}
    \notag
    \begin{aligned}
    A_0
    &= 
    \begin{pmatrix}
    \kappa_4 z & 0 \\[0.9mm]
    0 & 0
    \end{pmatrix},
    &
    A_{-1}
    &= -
    \begin{pmatrix}
    \kappa_1
    & 
    u
    \\[0.9mm]
    u v^2
    + \kappa_1 v
    + \kappa_2 u v 
    + \kappa_3
    & - u v + v u
    \end{pmatrix},
    &
    A_{-2}
    &= 
    \begin{pmatrix}
    v + \kappa_2
    & 
    - 1
    \\[0.9mm]
    v^2 + \kappa_2 v
    &
    - v
    \end{pmatrix},
    \end{aligned}
    \\[-1mm]
    \label{eq:Laxpair_P3D6case1}
    \\[-1mm]
    \notag
    \begin{aligned}
    B_1
    &= 
    \begin{pmatrix}
    \kappa_4 & 0 
    \\[0.9mm]
    0 & 0
    \end{pmatrix}
    ,
    &&&
    B_0
    &= z^{-1}
    \begin{pmatrix}
    u v + v u
    + \kappa_2 u 
    &
    - u
    \\[0.9mm]
    - \brackets{
    u v^2 
    + \kappa_1 v
    + \kappa_2 u v + \kappa_3
    }
    & 
    0
    \end{pmatrix}
    + z^{-1} (1 - \beta) [u, v] \, \mathbf{I}.
    \end{aligned}
\end{gather}

\subsection{The \PPainleve-2 system}
\hspace{3.8mm}
The Hamiltonian system corresponding to \eqref{hamPII} reads
\begin{align}
    \tag*{$\PIIn{\rm{H}}$}
    \label{eq:ncsysP2}
    \left\{
    \begin{array}{lcl}
         u'
         &=& - u^2 + v - \tfrac12 z,
         \\[2mm]
         v'
         &=& v u + u v + \kappa_3.
    \end{array}
    \right.
\end{align}
The degeneration $\text{\ref{eq:ncsysP4}} \to \text{\ref{eq:ncsysP2}}$ is given by the following replacements with the small parameter~$\varepsilon$:
\begin{gather}
\label{eq:P4toP2map}
    \begin{aligned}
    z 
    &= \tfrac14 \varepsilon^{-3} 
    - \varepsilon \, \tilde z,
    &&&
    u 
    &= - \tfrac14 \varepsilon^{-3} 
    - \varepsilon^{-1} \tilde u
    ,
    &&&
    v
    &= - 2 \varepsilon \tilde v,
    &&&
    \kappa_2
    &= - \tfrac{1}{16} \varepsilon^{-6},
    &&&
    \kappa_3
    &= 2 \tilde \kappa_3;
    \end{aligned}
    \\
    \label{eq:P4toP2map_la}
    \begin{aligned}
    &&&&
    \lambda
    &= -\tfrac18 \varepsilon^{-3} 
    + \tfrac12 \varepsilon^{-1} \tilde \lambda,
    &&&
    g
    &= 
    \begin{pmatrix}
    1 & 0 \\[0.9mm]
    - \varepsilon \, v & \varepsilon
    \end{pmatrix}.
    &&&&
    \end{aligned}
\end{gather}
The Lax pair \eqref{eq:matABform_P4} turns into
\begin{align} \label{eq:matABform_P2}
    \mathbf{A} (\lambda, z)
    &= A_2 \lambda^2 + A_1 \lambda + A_0,
    &
    \mathbf{B} (\lambda, z)
    &= B_1 \lambda + B_0,
\end{align}
with
\begin{equation} \label{eq:Laxpair_P20}
    \begin{gathered}
        \begin{aligned}
        A_2
        &= 
        \begin{pmatrix}
        2 & 0 \\[0.9mm] 0 & 0
        \end{pmatrix},
        &&&
        A_1
        &= 
        \begin{pmatrix}
        0 & -2 \\[0.9mm]
        - v & 0
        \end{pmatrix},
        &&&
        A_0
        &= 
        \begin{pmatrix}
        - v + z & - 2 u
        \\[0.9mm]
        u v + \kappa_3 & v
        \end{pmatrix},
        \end{aligned}
        \\[2mm]
        \begin{aligned}
        B_1
        &= 
        \begin{pmatrix}
        1 & 0 \\[0.9mm] 0 & 0
        \end{pmatrix},
        &&&
        B_0
        &= 
        \begin{pmatrix}
        - u & -1 \\[0.9mm]
        - \tfrac12 v & 0
        \end{pmatrix}.
        \end{aligned}
    \end{gathered}
\end{equation}

The degeneration of system \ref{eq:ncsysP3'}   to system \ref{eq:ncsysP2}   is carried out by the following substitution of the variables and parameters:
\begin{gather}
\label{eq:P3D6'toP2map}
\begin{gathered} 
    \begin{aligned}
    z 
    &= - \varepsilon^{-3}
    - 2 \varepsilon^{-1} \tilde z
    ,
    &&&
    u 
    &= 1 + 2 \varepsilon \, \tilde u
    ,
    &&&
    v
    &= \tfrac12 \varepsilon^{-1} \tilde v,
    \end{aligned}
    \\[1mm]
    \begin{aligned}
    \kappa_1
    &= \tfrac12 \varepsilon^{-3},
    &&&
    \kappa_2
    &= - \tfrac14 \varepsilon^{-3},
    &&&
    \kappa_3
    &= - \tfrac14 \varepsilon^{-3} \tilde \kappa_3,
    &&&
    \kappa_4
    &= \tfrac14.
    \end{aligned}
\end{gathered}
\end{gather}
To degenerate the Lax pairs, we make the transformation
\begin{align}
\label{eq:P3D6'toP2map_la}
    &&
    \lambda
    &= - 1 - 2 \varepsilon \, \tilde \lambda,
    &
    g
    &= 
    \begin{pmatrix}
    1 & 0 \\ 0 & 2 \varepsilon^2
    \end{pmatrix},
    &&
\end{align}
as a result of which the pair \eqref{eq:matABform_P3D6'_} turns into the pair \eqref{eq:matABform_P2} with coefficients given by the formula~\eqref{eq:Laxpair_P20}.

\subsection{The \PPainleve-1 system}
\hspace{3.8mm}
The Hamiltonian system related to \eqref{hamPI} is
\begin{align}
    \label{eq:ncsysP1}
    \tag*{$\PIn{\rm{H}}$}
    \left\{
    \begin{array}{lcl}
         u'
         &=& v,
         \\[2mm]
         v'
         &=& 6 u^2 + z.
    \end{array}
    \right.
\end{align}
This system can be obtained from the \ref{eq:ncsysP2} system by the limiting transition
\begin{gather}
\label{eq:P2toP1map}
\begin{gathered}
    \begin{aligned}
    z 
    &= - 6 \varepsilon^{-10}
    + \varepsilon^{2} \, \tilde z
    ,
    &&&
    u 
    &= \varepsilon^{-5} + \varepsilon \, \tilde u - \varepsilon^{4} \, \tilde v
    ,
    &&&
    v
    &= - 2 \varepsilon^{-10}
    + 2 \varepsilon^{-4} \, \tilde u
    - \varepsilon^{-1} \, \tilde v
    ,
    \end{aligned}
    \\[1mm]
    \begin{aligned}
    \kappa_3
    &= 4 \varepsilon^{-15}.
    \end{aligned}
\end{gathered}
\end{gather}
Supplementing the replacement \eqref{eq:P2toP1map} by the following gauge transformation:
\begin{align}
\label{eq:P2toP1map_la}
    &&
    \lambda
    &= - \varepsilon^{-5} + \varepsilon \, \tilde \lambda,
    &
    g
    &= e^{
    \varepsilon^{-5} \lambda^2
    + 4 \varepsilon^{-10} \lambda
    + \varepsilon^{-5} z
    }
    \begin{pmatrix}
    1 & 0 \\[0.9mm]
    \varepsilon^2 \, u & \varepsilon^2
    \end{pmatrix},
    &&
\end{align}
one can obtain from the Lax pair \eqref{eq:matABform_P2} a pair of the form
\begin{align} \label{eq:matABform_P1}
    \mathbf{A} (\lambda, z)
    &= A_2 \lambda^2 + A_1 \lambda + A_0,
    &
    \mathbf{B} (\lambda, z)
    &= B_1 \lambda + B_0,
\end{align}
where matrices are given by
\begin{equation} \label{eq:Laxpair_P10}
    \begin{gathered}
        \begin{aligned}
        A_2
        &= 
        \begin{pmatrix}
        0 & 0 \\[0.9mm] 
        2 & 0
        \end{pmatrix},
        &&&
        A_1
        &= 
        \begin{pmatrix}
        0 & -2 \\[0.9mm]
        - 2 u & 0
        \end{pmatrix},
        &&&
        A_0
        &= 
        \begin{pmatrix}
        - v & - 2 u
        \\[0.9mm]
        2 u^2 + z & v
        \end{pmatrix},
        \end{aligned}
        \\[2mm]
        \begin{aligned}
        B_1
        &= 
        \begin{pmatrix}
        0 & 0 \\[0.9mm] 1 & 0
        \end{pmatrix},
        &&&
        B_0
        &= 
        \begin{pmatrix}
        0 & -1 \\[0.9mm]
        - 2 u & 0
        \end{pmatrix}.
        \end{aligned}
    \end{gathered}
\end{equation}

\subsection{Special cases of \texorpdfstring{$\PIIIpr$}{P3'} type systems}\label{P3'D7} 
\hspace{3.8mm}

For completeness, let us consider  a special case $\PIIIpr (D_7)$ of non-abelian $\PIIIpr$ system with a non-polynomial Hamiltonian $H_3^{\prime\rm{H}}(D_7) = \tfrac1z \trace \brackets{H_3^{\prime}(D_7)}$, where 
\begin{align}
    \label{hamPIII'D7}
    H_3^\prime(D_7)
    &= \beta u^2 v^2 + (1 - \beta) u v u v
    + \kappa_2 u^2 v
    + \kappa_1 u v 
    + \kappa_3 u
    + \kappa_4 z u^{-1}.
\end{align}
This Hamiltonian leads to the system
\begin{gather}
    \label{eq:ncsysP3D7}
    \tag*{$\PIIIprn{\rm{H}}(D_7)$}
    \left\{
    \begin{array}{lcl}
         z \, u'
         &=& 2 u v u
         + \beta [u, [u, v]]
         + \kappa_1 u
         + \kappa_2 u^2
         ,
         \\[2mm]
         z \, v'
         &=& - 2 v u v 
         + \beta [v, [u, v]]
         - \kappa_1 v
         - \kappa_2 u v - \kappa_2 v u - \kappa_3
         + \kappa_4 z u^{-2}
    \end{array}
    \right.
\end{gather}
with the Laurent right hand side.

It is easy to check that for any parameter~$\beta$ there exists  a polynomial
$S_3^{(2)}$ satisfying Assumption \ref{ass:ncham}.   
 
The formulas
\begin{align}
    \label{eq:mapP3toP3D72}
    v
    &= \tilde v + \varepsilon^{-1} u^{-1},
    &
    \kappa_1
    &= \tilde \kappa_1 - 2 \varepsilon^{-1},
    &
    \kappa_3
    &= \tilde \kappa_3 - \varepsilon^{-1} \tilde \kappa_2,
    &
    \kappa_4
    &= \varepsilon \tilde \kappa_4
\end{align}
describe the degeneration $\text{\ref{eq:ncsysP3'}} \to \text{\ref{eq:ncsysP3D7}}$.  
Gauge transformation
\begin{align}
    &&
    \lambda
    &= - \tilde \lambda,
    &
    g
    &= \lambda^{- \varepsilon^{-1}} \,
    z^{- \varepsilon^{-1}}
    \begin{pmatrix}
    1 & 0 \\ - v & 1
    \end{pmatrix}
    ,
    &&
\end{align}
turns the Lax pair \eqref{eq:matABform_P3D6'_} into
\begin{align} \label{eq:matABform_P3D7'_}
    \mathbf{A} (\lambda, z)
    &= A_0 
    + A_{-1} \lambda^{-1} 
    + A_{-2} \lambda^{-2},
    &
    \mathbf{B} (\lambda, z)
    &= B_1 \lambda + B_0,
\end{align}
\begin{gather}
    \notag
    \begin{aligned}
    A_0
    &= 
    \begin{pmatrix}
    0 & 0 \\[0.9mm]
    \kappa_4 z u^{-1} & 0
    \end{pmatrix},
    &
    A_{-1}
    &= -
    \begin{pmatrix}
    u v + \kappa_1
    & 
    u
    \\[0.9mm]\kappa_2 u v + \kappa_3
    & - u v
    \end{pmatrix},
    &
    A_{-2}
    &= 
    \begin{pmatrix}
    - \kappa_2
    & 
    1
    \\[0.9mm]
    0
    &
    0
    \end{pmatrix},
    \end{aligned}
    \\[-1mm]
    \label{eq:Laxpair_P3D7case1}
    \\[-1mm]
    \notag
    \begin{aligned}
    B_1
    &= 
    \begin{pmatrix}
    0 & 0 
    \\[0.9mm]
    \kappa_4 u^{-1} & 0
    \end{pmatrix}
    ,
    &&&
    B_0
    &= z^{-1}
    \begin{pmatrix}
    u v 
    + \kappa_2 u 
    &
    - u
    \\[0.9mm]
    0
    & 
    u v
    \end{pmatrix}
    - z^{-1} \beta [u, v] \, \mathbf{I}.
    \end{aligned}
\end{gather}

For a degeneration $\text{\ref{eq:ncsysP3D7}} \to \text{\ref{eq:ncsysP1}}$, one can use the substitution
\begin{gather}
\label{eq:P3D72'toP1map}
\begin{gathered}
    \begin{aligned}
    z 
    &= 2 \varepsilon^{-10}
    + \varepsilon^{-6} \tilde z
    ,
    &&&
    u 
    &= 1 + \varepsilon^2 \, \tilde u
    ,
    &&&
    v
    &= 2 \varepsilon^{-5}
    + \varepsilon^{-2} \, \tilde v
    ,
    \end{aligned}
    \\[1mm]
    \begin{aligned}
    \kappa_1
    &= \varepsilon,
    &&&
    \kappa_2
    &= - 4 \varepsilon^{-5},
    &&&
    \kappa_3
    &= 12 \varepsilon^{-10},
    &&&
    \kappa_4
    &= 2.
    \end{aligned}
\end{gathered}
\end{gather}
Gauge transformation
\begin{align}
     &&
     \lambda
     &= 1 + \varepsilon^2 \, \tilde \lambda
     ,
     &
     g
     &= e^{2 \varepsilon^{-5} \lambda^{-1}}
     \begin{pmatrix}
     2 & 0 \\ \varepsilon^{4} v & \varepsilon^4
     \end{pmatrix}
     &&
\end{align}
leads to the Lax pair \eqref{eq:matABform_P1} with coefficients \eqref{eq:Laxpair_P10}.

\section{Degenerations of systems with a non-abelian constant}
\label{degdata_h}

\hspace{3.8mm}
Due to the appearance of the non-abelian constant $h$, the limiting transitions from Appendix \ref{sec:syshamlist} have to be modified by adding formulas for transforming the constant $h$ and slightly changing the gauge transformation. Below, without any comments, we present modified formulas that define the degenerations.

\medskip
\textbullet \,\,
$\PV \to \PIV$:
\begin{gather}
    \begin{aligned}
    z
    &= 1 + \sqrt{2} \varepsilon \, \tilde z,
    &&&
    u
    &= \tfrac{1}{\sqrt{2}} \varepsilon \, \tilde u,
    &&&
    v
    &= \tfrac{1}{\sqrt{2}} \varepsilon^{-1} \, \tilde v,
    \end{aligned}
    \\[2mm]
    \begin{aligned}
    \kappa_1
    &= \varepsilon^{-2},
    &&&
    \kappa_2
    &= - \tfrac12 \tilde \kappa_2,
    &&&
    \kappa_3
    &= \tfrac12 \varepsilon^{-2} \, \tilde \kappa_3,
    &&&
    h
    &= - \varepsilon^{-2} + \tfrac{1}{\sqrt{2}} \varepsilon^{-1} \, \tilde h,
    \end{aligned}
    \\[2mm]
    \begin{aligned}
    \lambda
    &= 1 + \sqrt{2} \varepsilon \, \tilde \lambda,
    &&&
    g
    &= e^{- \varepsilon^{-2} \lambda}
    \begin{pmatrix}
    1 & 0 \\ 0 & \sqrt{2} \varepsilon
    \end{pmatrix}
    .
    \end{aligned}
\end{gather}

\textbullet \,\,
$\PV \to \PIIIpr$:
\begin{gather}
    \begin{aligned}
    u
    &= 1 + \varepsilon \, \tilde u,
    &&&
    v
    &= \varepsilon^{-1} \, \tilde v,
    \end{aligned}
    \\[2mm]
    \begin{aligned}
    \kappa_1
    &= - \tilde \kappa_1 - \varepsilon^{-1} \, \tilde \kappa_2,
    &&&
    \kappa_2
    &= - \varepsilon^{-1} \tilde \kappa_2,
    &&&
    \kappa_3
    &= - \varepsilon^{-1} \, \tilde \kappa_3,
    &&&
    h
    &= \varepsilon \, \tilde h,
    \end{aligned}
    \\[2mm]
    \begin{aligned}
    \lambda
    &= \varepsilon^{-1} \, \tilde \lambda,
    &&&
    g
    &= 
    \begin{pmatrix}
    1 & 0 \\ 0 & \varepsilon
    \end{pmatrix}
    .
    \end{aligned}
\end{gather}

\textbullet \,\,
$\PIV \to \PII$:
\begin{gather}
    \begin{aligned}
    z
    &= \tfrac14 \varepsilon^{-3} 
    - \varepsilon \, \tilde z,
    &&&
    u
    &= - \tfrac14 \varepsilon^{-3}
    - \varepsilon^{-1} \, \tilde u
    ,
    &&&
    v
    &= - 2 \varepsilon \, \tilde v,
    \end{aligned}
    \\[2mm]
    \begin{aligned}
    \kappa_2
    &= - \tfrac{1}{16} \varepsilon^{-6},
    &&&
    \kappa_3
    &= 2 \tilde \kappa_3,
    &&&
    h
    &= 4 \varepsilon \, \tilde h,
    \end{aligned}
    \\[2mm]
    \begin{aligned}
    \lambda
    &= -\tfrac18 \varepsilon^{-3}
    + \tfrac12 \varepsilon^{-1} \, \tilde \lambda,
    &&&
    g
    &= e^{
    - \frac12 \varepsilon^{-3} \lambda
    + \frac14 \varepsilon^{-3} z
    }
    \begin{pmatrix}
    1 & 0 \\ - \varepsilon \, v & \varepsilon
    \end{pmatrix}
    .
    \end{aligned}
\end{gather}

\textbullet \,\,
$\PIIIpr \to \PII$:
\begin{gather}
    \begin{aligned}
    z
    &= - \varepsilon^{-3} 
    - 2 \varepsilon^{-1} \, \tilde z,
    &&&
    u
    &= 1 + 2 \varepsilon \, \tilde u,
    &&&
    v
    &= \tfrac{1}{2} \varepsilon^{-1} \, \tilde v,
    \end{aligned}
    \\[2mm]
    \begin{aligned}
    \kappa_1
    &= \tfrac12 \varepsilon^{-3},
    &&&
    \kappa_2
    &= - \tfrac14 \varepsilon^{-3},
    &&&
    \kappa_3
    &= - \tfrac14 \varepsilon^{-3} \, \tilde \kappa_3,
    &&&
    h
    &= \tfrac14 
    + \varepsilon^{2} \, \tilde h,
    \end{aligned}
    \\[2mm]
    \begin{aligned}
    \lambda
    &= - 1 - 2 \varepsilon \, \tilde \lambda,
    &&&
    g
    &= e^{
    - \frac14 \varepsilon^{-3} \lambda
    - \frac14 z
    }
    \begin{pmatrix}
    1 & 0 \\ 0 & 2 \varepsilon^2
    \end{pmatrix}
    .
    \end{aligned}
\end{gather}

\textbullet \,\,
$\PIIIpr \to \PIIIpr(D_7)$:
\begin{gather}
    \begin{aligned}
    v
    &= \tilde v 
    + \varepsilon^{-1} \, u^{-1},
    \end{aligned}
    \\[2mm]
    \begin{aligned}
    \kappa_1
    &= - 2 \varepsilon^{-1} 
    + \tilde \kappa_1,
    &&&
    \kappa_3
    &= - \varepsilon^{-1} \, \tilde \kappa_2 
    + \tilde \kappa_3,
    &&&
    h
    &= \varepsilon \, \tilde h,
    \end{aligned}
    \\[2mm]
    \begin{aligned}
    \lambda
    &= - \tilde \lambda,
    &&&
    g
    &= \lambda^{- \varepsilon^{-1}} \,
    z^{-\varepsilon^{-1}}
    \begin{pmatrix}
    1 & 0 \\ - v & 1
    \end{pmatrix}
    .
    \end{aligned}
\end{gather}

\textbullet \,\,
$\PIIIpr(D_7) \to \PI$:
\begin{gather}
    \begin{aligned}
    z
    &= 2 \varepsilon^{-10} 
    + \varepsilon^{-6} \, \tilde z,
    &&&
    u
    &= 1 
    + \varepsilon^{2} \, \tilde u,
    &&&
    v
    &= 2 \varepsilon^{-5}
    + \varepsilon^{-2} \, \tilde v,
    \end{aligned}
    \\[2mm]
    \begin{aligned}
    \kappa_1
    &= \varepsilon,
    &&&
    \kappa_2
    &= - 4 \varepsilon^{-5},
    &&&
    \kappa_3
    &= 12 \varepsilon^{-10},
    &&&
    h
    &= 2
    + \varepsilon^{4} \, \tilde h,
    \end{aligned}
    \\[2mm]
    \begin{aligned}
    \lambda
    &= 1 + \varepsilon^2 \, \tilde \lambda,
    &&&
    g
    &= e^{
    2 \varepsilon^{-5} \lambda^{-1}
    }
    \begin{pmatrix}
    2 & 0 \\ \varepsilon^4 \, v & \varepsilon^4
    \end{pmatrix}
    .
    \end{aligned}
\end{gather}

\textbullet \,\,
$\PII \to \PI$:
\begin{gather}
    \begin{aligned}
    z
    &= - 6 \varepsilon^{-10}
    + \varepsilon^2 \, \tilde z,
    &&&
    u
    &= \varepsilon^{-5}
    + \varepsilon \, \tilde u
    - \varepsilon^4 \, \tilde v
    ,
    &&&
    v
    &= - 2 \varepsilon^{-10}
    + 2 \varepsilon^{-4} \, \tilde u
    - \varepsilon^{-1} \, \tilde v
    ,
    \end{aligned}
    \\[2mm]
    \begin{aligned}
    \kappa_3
    &= 4 \varepsilon^{-15},
    &&&
    h
    &= \tfrac12 \varepsilon^{2} \, \tilde h,
    \end{aligned}
    \\[2mm]
    \begin{aligned}
    \lambda
    &= - \varepsilon^{-5} 
    + \varepsilon \, \tilde \lambda,
    &&&
    g
    &= e^{
    \varepsilon^{-5} \lambda^2
    - 2 \varepsilon^{-10} \lambda
    }
    \begin{pmatrix}
    1 & 0 \\ \varepsilon^2 \, u & \varepsilon^2
    \end{pmatrix}
    .
    \end{aligned}
\end{gather}

\bibliographystyle{plain}
\bibliography{bib}

\begin{thebibliography}{1}

\bibitem{Adler_Sokolov_2020}
V.~E. Adler and V.~V. Sokolov.
\newblock Non-{A}belian evolution systems with conservation laws.
\newblock {\em Mathematical Physics, Analysis and Geometry},
  24(1):\href{https://link.springer.com/article/10.1007/s11040--021--09382--6}{1--24},
  2021.
\newblock \href{https://arxiv.org/abs/2008.09174}{arXiv:2008.09174}.

\bibitem{Balandin_Sokolov_1998}
S.~P. Balandin and V.~V. Sokolov.
\newblock On the {P}ainlev{\'e} test for non-{A}belian equations.
\newblock {\em Physics letters A},
  246(3-4):\href{https://www.sciencedirect.com/science/article/abs/pii/S0375960198003363?via\%3Dihub}{267--272},
  1998.

\bibitem{Bobrova_Sokolov_2020_1}
I.~A. Bobrova and V.~V. Sokolov.
\newblock On matrix {P}ainlev{\'e}-4 equations. {P}art 1:
  {P}ainlev{\'e}-{K}ovalevskaya test.
\newblock {\em arXiv preprint
  \href{https://arxiv.org/abs/2107.11680}{arXiv:2107.11680}}, 2021.

\bibitem{Kawakami_2015}
H.~Kawakami.
\newblock Matrix {P}ainlev{\'e} systems.
\newblock {\em Journal of Mathematical Physics},
  56(3):\href{https://doi.org/10.1063/1.4914369}{033503}, 2015.

\bibitem{kontsevich1993formal}
M.~Konttsevich.
\newblock Formal (non)-commutative symplectic geometry, {T}he {G}elfand
  {M}athematical {S}eminars, 1990--1992.
\newblock {\em Fields Institute Communications, Birkh{\"a}user Boston}, pages
  \href{http://refhub.elsevier.com/S0393--0440(21)00176--5/bibB49C690771DD4516155FC02AE516406Ds1}{173--187},
  1993.

\bibitem{Odesskii_Sokolov_2021}
A.~Odesskii and V.~Sokolov.
\newblock Noncommutative elliptic {P}oisson structures on projective spaces.
\newblock {\em Journal of Geometry and Physics},
  169:\href{https://www.sciencedirect.com/science/article/pii/S0393044021001765?via\%3Dihub}{104330},
  2021.
\newblock \href{https://arxiv.org/abs/1911.03320}{arXiv:1911.03320}.

\bibitem{okamoto1980polynomial}
K.~Okamoto.
\newblock Polynomial {H}amiltonians associated with {P}ainlev{\'e} equations,
  {I}.
\newblock {\em Proceedings of the Japan Academy, Series A, Mathematical
  Sciences}, 56(6):\href{https://doi.org/10.3792/pjaa.56.264}{264--268}, 1980.

\bibitem{Retakh_Rubtsov_2010}
V.~S. Retakh and V.~V. Rubtsov.
\newblock Noncommutative {T}oda {C}hains, {H}ankel {Q}uasideterminants and
  {P}ainlev{\'e} {II} {E}quation.
\newblock {\em Journal of Physics. A, Mathematical and Theoretical},
  43(50):\href{https://iopscience.iop.org/article/10.1088/1751--8113/43/50/505204}{505204},
  2010.
\newblock \href{https://arxiv.org/abs/1007.4168}{arXiv:1007.4168}.

\end{thebibliography}

\end{document}